\documentclass[11pt, a4paper]{article}
\usepackage{amsmath}
\usepackage{amssymb}
\usepackage[dvips]{graphics}
\usepackage[dvips]{graphicx}
\usepackage{soul}
\usepackage{blkarray}

\usepackage{comment}
\usepackage{color} 
\usepackage{amsmath}
\usepackage{mathtools}
\usepackage{braket}
\usepackage[bookmarks=false]{hyperref} 
\usepackage{graphicx} 
\usepackage{multicol}
\usepackage{enumerate}
\usepackage{adjustbox}
\usepackage{pdflscape}
\usepackage{afterpage}
\usepackage{tikz}
\usetikzlibrary{matrix,arrows,positioning,intersections,calc,decorations.pathreplacing,decorations.markings}

\DeclareMathOperator{\Tr}{Tr}

\numberwithin{equation}{section}
\usepackage{color}

\newcommand{\eqa}{\begin{eqnarray}}
\newcommand{\eeqa}{\end{eqnarray}}
\newcommand{\beq}{\begin{equation}}
\newcommand{\eeq}{\end{equation}}
\newcommand{\nn}{\nonumber}

\newtheorem{theorem}{Theorem}

\newtheorem{definition}[theorem]{Definition}

\newtheorem{proposition}[theorem]{Proposition}

\newenvironment{proof}[1][Proof]{\textbf{#1.} }{\ \rule{0.5em}{0.5em}}

\newcommand{\benumerate}{\begin{enumerate}}
\newcommand{\eenumerate}{\end{enumerate}}

\newcommand{\bitemize}{\begin{itemize}}

\newcommand{\eitemize}{\end{itemize}}

\newcommand{\av}[2]{\langle #1, #2 \rangle }

\begin{document}

\title{Symmetric Matrix Ensemble and Integrable Hydrodynamic Chains}

\author{Costanza Benassi, Marta Dell'Atti, Antonio Moro\footnote{  costanza.benassi@northumbria.ac.uk, m.dellatti@northumbria.ac.uk, antonio.moro@northumbria.ac.uk}
 \\
\\
\it Department of Mathematics, Physics and Electrical Engineering \\ {\it Northumbria University Newcastle, UK} \\
}
\date{}

    \maketitle

\begin{abstract}
The partition function of the Symmetric Matrix Ensemble is identified with the $\tau-$function of a particular solution of the Pfaff Lattice. We show that, in the case of even power interactions, in the thermodynamic limit, the $\tau-$function corresponds to the solution of an integrable chain of hydrodynamic type. We prove that the hydrodynamic chain so obtained is diagonalisable and admits hydrodynamic reductions in Riemann invariants in an arbitrary number of components. \\

\end{abstract}


\section{Introduction}
\label{sec:intro}
Random Matrix Ensembles appear in relation to a variety of problems in mathematics and physics, often showing intriguing and unexpected connections. Symmetric, Hermitian and Symplectic ensembles  have been originally introduced to describe the statistics of energy levels of heavy nuclei and complex systems  \cite{Dysoni,wigner_1951, porter_rosenzwieg}; the Circular Unitary Ensemble and the zeros of the Riemann Zeta function, respectively, in the thermodynamic limit and in the far limit on the critical line appear to follow the same statistics \cite{Keating00randommatrix}; the Hermitian Matrix Ensemble (HME) arises from discrete approximations of Topological Field Theory~\cite{Witten, Kontsevich:1992ti, Bessis:1980ss, Brezin:1977sv}; the partition function of the orthogonal ensemble appears in the calculation of the generating function for specific subsequences of permutations \cite{ADLER2004190}.
In addition to the above mentioned connections, a remarkable result is the identification of the partition function of Random Matrix Ensembles with particular solutions of nonlinear integrable systems \cite{Witten, Kontsevich:1992ti,vanMoerbeke2000,vanMoerbekenotes}. For example, the partition function for the HME can be identified simultaneously with the $\tau-$function of a particular solution of the Toda Lattice hierarchy and of the Kadomtev-Petviashvili hierarchy~\cite{vanMoerbeke1995}. 

In the present paper, we propose an approach to the study of the Symmetric Matrix Ensemble (SME) based on the method of differential identities developed and effectively applied to a variety of statistical mechanical models (see e.g. \cite{HermitianPRE,MoroPotts,Moro2014,MoroAnnals}). In the case of SMEs, the method relies on its underlying integrable structure realised by the Pfaff Lattice \cite{JimboMiwa,vanMoerbeke1999, TakasakiPfaff,TakasakiFay,Zabrodin}. In the thermodynamic limit, such integrable structure allows the derivation of a system of partial differential equations (PDEs)of dispersionless type for the order parameters. We note that a first direct connection between the dispersionless limit of the Pfaff Hierarchy and  the thermodynamic limit of matrix models has been studied in \cite{KodamaPierce}. In this paper, we prove that for the SME with even interactions, the order parameters satisfy an integrable hydrodynamic chain of PDEs. Integrable hydrodynamic chains, of which the moments Benney chain is the protypical example \cite{Benney}, represent an important class of integrable systems of dispersionless type that has attracted a great deal of interest over the last two decades in relation to their classification, construction of new integrable systems via reductions and associated Hamiltonian structures \cite{chain,PavlovChains,PavlovChainsClass,PavlovChainsKup,ChainsOdesskiSokolov}. We find, at the best of our knowledge, a new example of integrable hydrodynamic chain and prove its integrability via the method of hydrodynamic reductions \cite{FKred,chain}.

The SME is characterised by the partition function of the form
\begin{equation}
Z_n(\mathbf{t})=\int_{\mathcal{S}_n} \text{e}^{-H(M)}\, dM,
\label{eq:partition_n_harr_intro}
\end{equation}
where $dM$ is the Haar measure, i.e. $dM \coloneqq \prod_{1\le i\le j \le n} dM_{ij}$, and $\mathcal{S}_n$ is the set of real symmetric matrices. The function $H(M)$, chosen as
\[
H(M) = - \Tr \left( -\,\frac{M^2}{2}\,+\,\sum_{k\geq 1} t_k M^k \right),
\]
is referred to as the Hamiltonian and parameters $\mathbf{t} = \{t_{k}\}_{k\geq 1}$ are the coupling constants. This terminology refers to Matrix Models of interest in Quantum Field Theory where $H(M)$ is interpreted as the Hamiltonian of the system and $t_{k}$ are the coupling constants associated to different degrees of interaction \cite{Witten}. Hence, the free particle theory  corresponds to the case where all coupling constants $t_{k}$ vanish. Notice that for $\mathbf{t} = \mathbf{0}$, the expression \eqref{eq:partition_n_harr_intro} reduces the partition function of the Gaussian Orthogonal Ensemble (GOE). 

Based on a classical result by Weyl \cite{Weyl}, observing that $H(M)$ depends on the eigenvalues~$\{ z_{k}\}_{k=1}^{n}$ of $M$ only,  the integral \eqref{eq:partition_n_harr_intro} can be reduced to an integral over the eigenvalues of the form
\begin{equation}
Z_n(\mathbf{t})= C_n \int_{\mathbb{R}^n} |\Delta_n(z)| \, \prod_{i=1}^n \text{e}^{- \frac{z^2_i}{2} + \sum_{k\geq 1}t_k\,z_i^k}\,dz_i \,
\label{eq:partition_n_eigenvalues_intro}
\end{equation}
where $C_n$ is a constant obtained from the integration over the remaining degrees of freedom, and $\Delta_n(z)$ denotes the Vandermonde determinant
$ \Delta_n(z)=\prod_{1\leq i < j \leq n}(z_i-z_j)$. 
A fundamental result by Adler and van Moerbeke \cite{vanMoerbeke1999} establishes that for $2n \times 2n$ symmetric matrices, the function
\begin{equation}
\tau_{2 n}(\mathbf{t}) := \frac{1}{(2n)! C_{2n}}Z_{2n}(\mathbf{t}) = \frac{1}{(2n)!}\int_{\mathbb{R}^{2n}} |\Delta_{2n}(z)| \, \prod_{i=1}^{2n} \text{e}^{- \frac{z^2_i}{2} + \sum_{k\geq 1}t_k\,z_i^k}\,dz_i \, 
\label{eq:partition_and_tau_intro}
\end{equation}
is the Pfaffian of the moments matrix
\[
m_{2n}(\mathbf{t})=\left (\av{x^{i}}{y^{j}}_\mathbf{t}\right )_{0 \le i,j < 2n-1}
\]
where $\av{\,\cdot\,}{\cdot\,}_{\mathbf{t}}$  is a skew-symmetric scalar product induced by the measure on the SME. More specifically
$\tau_{2n}(\mathbf{t}) = \textup{pf} \left( m_{2n} \right)$
is a particular solution of the Pfaff Lattice, an integrable system arising in relation to the algebra splitting of $\textup{gl}(\infty)$ into $\textup{sp}(\infty)$ and the algebra of $2\times 2$ blocks lower triangular matrices \cite{vanMoerbeke1999}. The Pfaff Lattice equations are constructed based on the following unique factorisation of the semi-infinite moments matrix
\begin{equation}
m_{\infty}(\mathbf{t})= \left( Q(\mathbf{t})^{-1} \right) \,J\, \left( Q(\mathbf{t})^{-1}\right)^{\top}  
\label{skew_decomposition_intro}
\end{equation} 
where $J$ is the semi-infinite skew-symmetric matrix such that $J^2=-I$ and $Q$ is a semi-infinite lower triangular matrix. The Lax matrix of the form
\begin{equation}
L(\mathbf{t})= Q(\mathbf{t}) \Lambda Q(\mathbf{t})^{-1}, \,
\label{eq:l_t_intro}
\end{equation} 
where $\Lambda = \{\delta_{i,j-1} \}_{i,j=1}^{\infty}$ is the shift matrix with $\delta_{i,j}$ the Kronecker delta, satisfies the Lax equation~\cite{vanMoerbeke1999}
\begin{equation}
\frac{\partial L}{\partial t_{k}} = \left[ \,- (L^{k})_{\mathfrak{t}}\,, L \, \right]  \,.
\label{eq:hamiltonian_L_intro}
\end{equation}
The projection $\left (A \right)_{\mathfrak{t}}$ for a given matrix $A$ is defined as follows
\begin{equation}
A_{\mathfrak{t}}:= A_{-} - J\,(A_{+})^{\top}\,J + \frac{1}{2} \left( A_{0} - J (A_0)^{\top} J, \right)\label{eq:projection}
\end{equation}
with $A_{\pm}$ denoting, respectively, the upper and lower triangular part of $A$, with all $2\times 2$ diagonal blocks equal to zero, and $A_{0}$ the block diagonal part of $A$ with $2\times 2$ diagonal blocks. The entries of the Lax matrix $L$ depend on the sequence of $\tau-$funtions $\{\tau_{2n}(\mathbf{t})\}_{n\geq1}$ and their derivatives with respect to the coupling constants $t_{k}$. The Lax matrix associated to the SME partition function is a solution of the Lax equation with initial condition specified by $\tau_{2n}(\mathbf{t})$  and its derivatives evaluated at $\mathbf{t} = \mathbf{0}$. It is important to note that in this case the integrals are specified by the Gaussian measure and can be evaluated using Selberg's theorem (see e.g. \cite{Meh2004}). The study of the form of the Lax equation \eqref{eq:hamiltonian_L_intro} and its asymptotic properties in the large $n$ limit provides important information on the generic properties of solutions, as for example their singularities and breaking mechanisms, independently of the particular initial condition. More specifically,
we observe that the components of the Lax equation \eqref{eq:hamiltonian_L_intro} can be organised as two coupled systems of ODEs, a {\it double chain} in infinite components, of the form
\begin{gather}
\begin{aligned}
\partial_{t_{k}}\mathbf{v}_{n} &= F_{k}[\mathbf{v},\mathbf{w}]  \\
\partial_{t_{k}}\mathbf{w}_{n} &= G_{k}[\mathbf{v},\mathbf{w}]
\end{aligned}
\label{explicit_full_lattice}
\end{gather}
where $\mathbf{v}_{n}$ and $\mathbf{w}_{n}$ are the entries of the Lax matrix $L$ suitably recast in the form of infinite vectors, e.g. $\mathbf{v}_{n} = (\dots,v^{-k}_{n},\dots,v^{-1}_{n},v^{0}_{n},v^{1}_{n},\dots,v^{k}_{n},\dots)^{\top}$, each associated to a position $n$ on the lattice. $F_{k}$ and $G_{k}$ are nonlocal functions on the lattice, evaluated on specific subsets of sites that depend on the chosen $t_{k}-$flow.

We then proceed with the study of the Lax equations for SMEs with even power interactions such that the partition function is of the form
\begin{equation}
Z_{2n}(\mathbf{t})=\int_{\mathcal{S}_{2n}} \text{e}^{\Tr \left( -\,\frac{M^2}{2}\,+\,\sum_{k \ge 1} t_{2k} M^{2k} \right)}\, dM. 
\label{eq:partition_n_even_intro}
\end{equation}
The above choice automatically selects a reduction of the even Pfaff Lattice given by the hierarchy~\eqref{eq:hamiltonian_L_intro} restricted to the even times $t_{2k}$. 
Hence, the system \eqref{explicit_full_lattice} is replaced by a single chain of the form
\begin{gather}
\begin{aligned}
\partial_{t_{2 k}}\mathbf{w}_{n} &= H_{k}[\mathbf{w}]
\end{aligned}
\label{explicit_reduced_lattice}
\end{gather}
where, similarly to the more general case \eqref{explicit_full_lattice}, $H_{k}[\mathbf{w}]$ is a nonlocal function on the lattice. Introducing the variable $x = \varepsilon n$ and the interpolation function
\begin{equation*}
\mathbf{u}(x) := \mathbf{w}\left (\frac{x}{\varepsilon} \right)
\end{equation*}
such that $\mathbf{u}(x\pm j \varepsilon) = \mathbf{w}_{n \pm j}$ for some integer $j$, the thermodynamic limit of the matrix ensemble, i.e. the limit for $n \to \infty$, corresponds to the continuum limit of the reduced even Pfaff Hierarchy~\eqref{explicit_reduced_lattice} obtained by taking $\varepsilon \to 0$ such that $x = \varepsilon n$ remains finite and $\mathbf{u}(x)$ is an infinite component vector field of the continuous variable $x$. Substituting the interpolating function in \eqref{explicit_reduced_lattice} and expanding in Taylor series for $\varepsilon \to 0$, at the leading order, one obtains a hierarchy of compatible partial differential equations. A direct calculation performed for the first flows associated to $t_{2}$, $t_{4}$ and $t_{6}$ shows that the equations so obtained constitute an infinite system of first order PDEs of hydrodynamic type, referred to as hydrodynamic chains, of the form
\begin{equation}
\label{hydromatrixchain_intro}
{\bf u}_{T_{2 k}} = A^{(2 k)}({\bf u}) \, {\bf u}_{x} 
\end{equation}
where $T_{2k}$ corresponds to the time variable $t_{2k}$ suitably rescaled, e.g. $T_2 = \varepsilon t_2$. $A^{(2k)}(\mathbf{u})$ is a sparse matrix where each row contains a finite number of elements depending on a finite number of components of the vector field ${\mathbf{u}(x)}$. Infinite matrices of this type are referred to as  class $C$ (chain-class) matrices~\cite{chain}. We conjecture that the form \eqref{hydromatrixchain_intro} holds for all $k=1,2,\dots$\, . We do not have at the moment a proof of this conjecture.

We prove that the hydrodynamic chain so obtained passes the diagonalisability test (see Proposition~\ref{th:Htest}) introduced in \cite{chain} and is integrable in the sense of hydrodynamic reductions (Theorem \ref{th:hrydrored}) \cite{FKred,chain}. As the compatibility of the hierarchy \eqref{hydromatrixchain_intro} implies that the matrices $A^{(2k)}$ commute, it is sufficient to perform the diagonalisability test for the first flow of the hierarchy only.
The test establishes that a hydrodynamic chain of class $C$ defined by the infinite matrix $A(\mathbf{u}) = \{a^{i}_{j} \}_{i,j=-\infty}^{\infty}$ is diagonalisable if and only if all components of the Haantjes tensor 
\begin{equation}
H^{i}_{jk} = N^{i}_{pr} \,  a^{p}_{j} \,  a^{r}_{k} - N^{p}_{jr} \,  a^{i}_{p} \,  a^{r}_{k} - N^{p}_{rk} \,  a^{i}_{p} \,  a^{r}_{j} + N^{p}_{jk} \, a^{i}_{r} \, a^{r}_{p},
\label{eq:haantjes_intro}
\end{equation}
where $N^i_{jk}$ is the Nijenhuis tensor
\begin{equation}
N^{i}_{jk} = a^{p}_{j} \,  \partial_{u^{p}} a^{i}_{k} -a^{p}_{k} \, \partial_{u^{p}} a^{i}_{j} - a^{i}_{p} \left(\partial_{u^{j}} a^{p}_{k} -\partial_{u^{k}} a^{p}_{j} \right)\label{eq:Ndef_intro}
\,,
\end{equation}
vanish identically.
The notion of integrability in the sense of hydrodynamic reductions for a hydrodynamic chain  extends the similar concept introduced in the context of finite component systems \cite{FKred} and characterises the chain under consideration as integrable if it admits $N-$phase solutions of the form $\mathbf{u}(R^{1},\dots,R^{N})$ for any integer $N$, where $R^{i} = R^{i}(x,t)$ ({\it Riemann invariants}) satisfy the semi-Hamiltonian diagonal system of hydrodynamic type
\begin{equation}
    R^{i}_{t} = \lambda^{i}\left(R^{1},\dots,R^{N}\right) R^{i}_{x} \qquad i=1,\dots,N.
\label{Rinv_intro}
\end{equation}
The time $t$ can be identified, subject to a suitable re-scaling, with any of the ``times" $t_{2k}$ and $\lambda^{i}\left(R^{1},\dots,R^{N}\right)$ are the characteristic speeds of the corresponding flow. The system for Riemann invariants \eqref{Rinv_intro} is required to fulfill the semi-Hamiltonian property which can be expressed in terms of the following differential constraint on the characteristic speeds
\begin{equation}
\label{semiH_intro}
\partial_{k} \left (\frac{\partial_{j} \lambda^{i}}{\lambda^{j} - \lambda^{i}} \right) = \partial_{j} \left (\frac{\partial_{k} \lambda^{i}}{\lambda^{k} - \lambda^{i}} \right), \qquad i \neq j \neq k,
\end{equation}
with the notation $\partial_{i} = \partial / \partial R^{i}$.
The condition \eqref{semiH_intro} guarantees that  equation \eqref{Rinv_intro}  constitutes a system of conservation laws \cite{Sevennec}. A classical result by Tsarev establishes that the system~\eqref{Rinv_intro} is completely integrable by the generalised hodograph method (\cite{Tsarev}, see also \cite{1989RuMaS..44...35D}) and the solution is given by the following equation
\[
x + \lambda^{i}\left(R^{1},\dots,R^{N}\right) t = \mu^{i}\left(R^{1},\dots,R^{N}\right) \qquad i =1,\dots,N
\]
where the functions $\mu^{i}\left(R^{1},\dots,R^{N}\right)$  satisfy the system of linear PDEs of the form
\begin{equation}
\frac{\partial_{j} \lambda^{i}}{\lambda^{i} - \lambda^{j}} = \frac{\partial_{j} \mu^{i}}{\mu^{i} - \mu^{j}}.
\label{Rsymm}
\end{equation}
The solution to the system \eqref{Rsymm} is parametrised via $N$ functions of one variable that can be fixed by the initial conditions on the functions $R^{i}$. 

The paper is organised as follows. In Section \ref{sec:SMEreview}, we review some results regarding SMEs and their relationship with the Pfaff Lattice (see e.g. \cite{vanMoerbeke1999, vanMoerbekenotes}). In Section \ref{sec:pfaffgen}, we study the structure of the Lax matrix \eqref{eq:l_t_intro} and write the explicit evolution equations for the first flow of the Pfaff Lattice. Section \ref{sec:paffeven} is devoted to the SME with even degree interactions and its relation with the even Pfaff Hierarchy. In Section \ref{sec:limit} we study the thermodynamic limit and show that the resulting hydrodynamic chain is diagonalisable and integrable. We then conclude with some final remarks in Section \ref{sec:conclusions}. 
Appendices provide some addition technical details as well as elements that are subject of further studies. The expressions for the second flow of the Pfaff Lattice with odd and even times are provided in Appendix \ref{app:t2flow_all}; the explicit form of the coupled system \eqref{explicit_full_lattice} and its higher order corrections are given in Appendix~\ref{app:dispfaff}; higher order corrections to the hydrodynamic chain \eqref{explicit_reduced_lattice} are reported in Appendix \ref{app:even_corrections}; Appendix \ref{sec:N} lists non-zero entries of the Nijenhuis tensor \eqref{eq:Ndef_intro} for the hydrodynamic chain \eqref{explicit_reduced_lattice}.

\section*{Acknowledgements}
This work is supported by the Leverhulme Trust RPG 2017-228 (PI A.M.). Authors also thank the London Mathematical Society, the Royal Society International Exchanges Grant IES-R2-170116 (PI A.M.), GNFM - Gruppo Nazionale per la Fisica Matematica, INdAM (Istituto Nazionale di Alta Matematica) for supporting activities that contributed to the research reported in this paper.

\section{Symmetric Matrix Ensemble and Pfaff Lattice}
\label{sec:SMEreview}

In this section we briefly review definitions and properties of the SME and its connection with the Pfaff Lattice  with a focus on aspects that are relevant for the purposes of this paper \cite{vanMoerbeke1999, vanMoerbeke2002, vanMoerbekenotes}.
As mentioned above, the partition function \eqref{eq:partition_n_harr_intro}  can be reduced, up to a proportionality constant, to the integral over the eigenvalues of the form 
\[
\tau_{2 n}(\mathbf{t}) := \frac{1}{(2n)!}\int_{\mathbb{R}^{2n}} |\Delta_{2n}(z)| \, \prod_{i=1}^{2n} \text{e}^{- \frac{z^2_i}{2} + \sum_{k\geq 1}t_k\,z_i^k}\,dz_i \,,
\]
where $\tau_{2n}$ is referred to as Pfaffian $\tau$-function. The function $\tau_{2n}$ is in fact the Pfaffian of the $\mathbf{t}-$dependent moment matrix $m_{2n}(\mathbf{t})=\big(\mu_{ij}(\mathbf{t})\big)_{0 \le i,j < 2n-1}$, i.e.
\begin{equation}
\tau_{2n}(\mathbf{t}) = \text{pf}(m_{2n}(\mathbf{t}))=\left( \text{det} \, m_{2n}(\mathbf{t}) \right)^{1/2} \,
\label{eq:tau_pfaffian}
\end{equation}  
where entries of $m_{2n}(\mathbf{t})$ are constructed via the skew-symmetric inner product $\langle \cdot \,,\,\cdot \rangle_{\mathbf{t}} $ 
\begin{equation}
\mu_{ij}(\mathbf{t}) = \langle x^i\,,\,y^j \rangle_{\mathbf{t}} := \int \int_{\mathbb{R}^2} x^i \, y^j \,\sigma(x-y)\,\text{e}^{-\frac{1}{2} (x^2 + y^2) + \sum_{k}t_k\,(x^k+y^k)}\,dx\,dy
\label{eq:moments}
\end{equation}
with $\sigma(x)=\text{sign}(x)$. Noting that $\mu_{ij}=-\mu_{ji}$, the moments matrix $m_{2n}$ is skew-symmetric and takes the form
\begin{equation}
m_{2n}=\big( \mu_{i\,j} \big)_{0\le i,j\le 2n-1}= \begin{pmatrix}
0 & \mu_{0\,1} & \mu_{0\,2} & \mu_{0\,3} & \mu_{0\,4} & \mu_{0\,5} & \dots \\[1ex] 
-\mu_{0\,1} & 0 & \mu_{1\,2} & \mu_{1\,3} & \mu_{1\,4} & \mu_{1\,5} & \dots \\[1ex] 
-\mu_{0\,2} & -\mu_{1\,2} & 0 & \mu_{2\,3} & \mu_{2\,4} & \mu_{2\,5} & \dots\\[1ex] 
-\mu_{0\,3} & -\mu_{1\,3} & -\mu_{2\,3} & 0 & \mu_{3\,4} & \mu_{1\,5} & \dots\\[1ex] 
-\mu_{0\,4} & -\mu_{1\,4} & -\mu_{2\,4} & -\mu_{3\,4} & 0 & \mu_{4\,5} & \dots\\[1ex] 
-\mu_{0\,5} & -\mu_{1\,5} & -\mu_{2\,5} & -\mu_{3\,5} & -\mu_{4\,5} & 0 & \dots \\[1ex]
\vdots & \vdots & \vdots & \vdots & \vdots & \vdots & \ddots \\ 
\end{pmatrix}
\end{equation} 
The evolution equations of matrix elements $\mu_{ij}$ with respect to the coupling constants $\{t_k\}_{k\in \mathbb{N}}$ follow from the direct differentiation from the definition \eqref{eq:moments} and read as
\begin{equation}
\frac{\partial \mu_{ij}}{\partial t_k}=\mu_{i+k,\,j}+\mu_{i,\,j+k} .
\label{lattice_moments}
\end{equation}
Therefore, the semi-infinite moment matrix $m_{\infty}$  satisfies the equation
\begin{equation}
\frac{\partial m_{\infty}}{\partial t_k} = \Lambda^k\,m_{\infty}+m_{\infty} \, \Lambda^k \,, 
\end{equation}
where $\Lambda$ is the shift matrix 
\begin{equation}
\Lambda=\begin{pmatrix}
\hspace{1ex} 0 & 1 & 0 & 0 & \dots \hspace{1ex} \\[1ex]
0 & 0 & 1 & 0 & \dots \\[1ex]
0 & 0 & 0 & 1 & \dots \\[1ex]
\vdots & \vdots & \vdots & \vdots & \ddots \\[1ex] 
\end{pmatrix}.\\[1ex] 
\end{equation}
As mentioned above regarding the equation \eqref{skew_decomposition_intro}, the matrix $m_{\infty}$ admits the unique factorisation~\cite{vanMoerbeke1995, vanMoerbeke2002, vanMoerbekenotes}
\begin{equation}
\label{minfactor}
m_{\infty}(\mathbf{t})= \left( Q(\mathbf{t})^{-1} \right)\,J\,\left( Q(\mathbf{t})^{-1}\right)^{\top}  
\end{equation}
where $J$ is the semi-infinite skew-symmetric matrix such that $J^2=-I$ and $Q$ is a semi-infinite lower triangular matrix of the form 
\begin{equation}
Q(\mathbf{t})= \begin{tikzpicture}[>=stealth,thick,baseline]
    \matrix [matrix of math nodes,row sep=.4em, column sep=1.4em,left delimiter=(,right delimiter=)](A){ 
 \ddots & 0  & 0 & 0 & 0 & \dots \\[1ex] 
 & Q_{2n,2n} & 0 & 0 & 0 & \dots \\[1ex] 
 & 0 & Q_{2n,2n} & 0 & 0 & \dots \\[1ex] 
 & \ast & \ast & Q_{2n+2,2n+2} & 0 & \dots \\[1ex] 
 & \ast & \ast  & 0 & Q_{2n+2,2n+2} & \dots \\[1ex] 
 & \vdots & \vdots & \vdots & \vdots & \ddots  \\[1ex] 
   };
\draw[opacity=0.7]  ($(A-2-2)+(-20pt,10pt)$) rectangle ($(A-3-3)+(20pt,-10pt)$);
\draw[opacity=0.7]  ($(A-4-4)+(-30pt,10pt)$) rectangle ($(A-5-5)+(30pt,-10pt)$);
\end{tikzpicture}.    \\[1ex] 
\label{eq:q_t}
\end{equation}
Due to the factorisation \eqref{minfactor}, elements of the matrix $Q$ depend on the moments $\mu_{ij}$ as well as the Pfaffian $\tau-$function $\tau_{2n}$. From equation \eqref{lattice_moments}, it follows that the evolution of the Pfaffian $\tau-$function with respect to the $k-$th time can be written as
\begin{equation}
\frac{\partial \tau_{2n}}{\partial t_k} = \sum_{i,j=0}^{2n-1} \frac{\partial \tau_{2n}}{\partial \mu_{i,j}} \, \frac{\partial \mu_{i,j}}{\partial t_k} =    \sum_{i,j=0}^{2n-1} \frac{\partial \tau_{2n}}{\partial \mu_{i,j}} \, \left(  \mu_{i+k,j} + \mu_{i,j+k} \right).
\end{equation}
Therefore, elements of the decomposition matrix $Q$ are expressed in terms of $\tau_{2n}$ and suitable combinations of its derivatives with respect to the times $\{t_k\}_{k\in\mathbb{N}}$ determined by the Schur's polynomials of the differential operators $\{\partial_{t_k}\}_{k\in\mathbb{N}}$ \cite{vanMoerbekenotes}.

The factorisation of the moments matrix allows to define the Lax matrix 
\begin{equation}
L(\mathbf{t})= Q(\mathbf{t}) \Lambda Q(\mathbf{t})^{-1}
\label{Lmatdef}
\end{equation}
for which the following theorem holds:

\newpage
\begin{theorem}[\cite{vanMoerbeke1999}] The function $\tau_{2n}$ is a $\tau$-function for the Pfaff Lattice, i.e. the following operator 
\begin{equation}
L(\mathbf{t})= Q(\mathbf{t}) \Lambda Q(\mathbf{t})^{-1}  = \begin{pmatrix}
0  & 1 & 0 & 0 & 0 & 0 & \dots \\[1.5ex]
\hspace{1ex} \ast \hspace{1ex} & \partial_{t_1} \log \tau_2  &   \left( \dfrac{\tau_4 \, \tau_0}{\tau_2^2} \right)^{1/2} & 0 & 0 & 0 & \ddots  \\[4ex]
\hspace{1ex} \ast \hspace{1ex} & \hspace{1ex} \ast \hspace{1ex}   & \hspace{-2ex} -\partial_{t_1} \log \tau_2 & 1 & 0 & 0 & \ddots \\[1.5ex]
\hspace{1ex} \ast \hspace{1ex}  & \hspace{1ex} \ast \hspace{1ex}  & \hspace{1ex} \ast \hspace{1ex}  & \hspace{-2ex} \partial_{t_1} \log \tau_4  &    \left( \dfrac{\tau_6 \, \tau_2}{\tau_4^2} \right)^{1/2} & 0 & \ddots  \\[3ex]
\hspace{1ex} \ast \hspace{1ex}  & \hspace{1ex} \ast \hspace{1ex}  & \hspace{1ex} \ast \hspace{1ex}  &  \hspace{1ex} \ast \hspace{1ex}    &    -\partial_{t_1} \log \tau_4 & 1 & \ddots  \\[3ex]
\vdots & \ddots  & \ddots & \ddots & \ddots &  \ddots  & \ddots \\[1.5ex]
\end{pmatrix} 
\label{eq:l_t}
\end{equation} 
satisfies the commuting equations
\begin{equation}
\frac{\partial L}{\partial t_k} = \left[ \,- (L^k)_{\mathfrak{t}}\,, L \, \right]  \, \qquad k\in \mathbb{N},
\label{eq:hamiltonian_L}
\end{equation}
The equation \eqref{eq:hamiltonian_L} is referred to as the Lax equation of the Pfaff Lattice.
\end{theorem}

The star symbols $\ast$ in the expression of the Lax matrix \eqref{eq:l_t} stand for suitable differential expressions of the $\tau-$functions $\tau_{2k}$ and $A_{\mathfrak{t}}$ is the projection  defined in 
\eqref{eq:projection}.
This follows from the splitting of the Lie algebra $\text{gl}(\infty)$ 
\begin{equation}
    \text{gl}(\infty) = \mathfrak{t} \oplus \mathfrak{n} \begin{cases} \mathfrak{t} = \{\text{lower triangular matrices of the form \eqref{eq:q_t}}\} \\[1.5ex] 
    \mathfrak{n} = \text{sp}(\infty) = \{ A \text{ such that } J A^{\top} J = A\}
    \end{cases}
\end{equation}
which yields the unique decomposition 
\begin{equation}
    \begin{split}
        A & = A_{\mathfrak{t}} + A_{\mathfrak{n}} \\[1.5ex]
        & = A_{-} - J\,(A_{+})^{\top}\,J + \frac{1}{2} \left( A_{0} - J (A_0)^{\top} J  \right)  + A_{+} + J\,(A_{+})^{\top}\,J + \frac{1}{2} \left( A_{0} + J (A_0)^{\top} J  \right) \, 
    \end{split}
\end{equation}
where $A_{\pm}$ denote, respectively, the upper and lower triangular part of $A$, with all $2\times 2$ diagonal blocks equal to zero, and $A_{0}$ the block diagonal part of $A$ with $2\times 2$ diagonal blocks. 
\newpage

\section{Lattice equations and initial conditions for the Pfaff Hierarchy}
\label{sec:pfaffgen}
In this section we further investigate the structure of the Lax equation \eqref{eq:hamiltonian_L}. Our main observation is that the Lax equation can be recast in the form of a two-component infinite chain.

Let us introduce the following notation for the Lax matrix \eqref{eq:l_t} 

\begin{equation}
L=\begin{pmatrix}
\hspace{2ex} 0 & 1 & 0 & 0 & 0 & 0 & 0 & 0 & 0 & \dots \hspace{2ex} \\[2ex]
\hspace{2ex}w_1^{-1} & v^{0}_{1} & w_1^0 & 0 & 0 & 0 & 0 & 0 & 0 & \dots \hspace{2ex} \\[2ex]
\hspace{2ex} v^{-1}_{1} & w_1^{1} & -v^{0}_{1} & 1 & 0 & 0 & 0 & 0 & 0 & \dots \hspace{2ex} \\[2ex]
\hspace{2ex} w_1^{-2} & v^{1}_{1} & w_2^{-1} & v^{0}_{2} &  w_2^0& 0 & 0 & 0 & 0 & \dots \hspace{2ex} \\[2ex]
\hspace{2ex} v^{-2}_{1} & w_1^2 & v^{-1}_{2} & w_2^1 & -v^{0}_{2} & 1 & 0 & 0 & 0 & \dots \hspace{2ex} \\[2ex]
\hspace{2ex} w_1^{-3} & v^{2}_{1} & w_2^{-2} & v^{1}_{2} & w_3^{-1} & v^{0}_{3} & w_3^0 & 0 & 0 & \dots \hspace{2ex} \\[2ex]
\hspace{2ex} v^{-3}_{1} & w_1^3 & v^{-2}_{2} & w_2^{2} & v^{-1}_{3} & w_3^1 & -v^{0}_{3} & 1 & 0 & \dots \hspace{2ex} \\[2ex]
\hspace{2ex} \vdots & \ddots & \ddots & \ddots & \ddots & \ddots & \ddots & \ddots & \ddots & \ddots \hspace{2ex} \\[2ex]
\end{pmatrix}.
\vspace{2ex} 
\end{equation}
The variables $\{w^{0}_{n}\}_{n\geq 1}$ constitute the non-constant entries in the first upper diagonal (even positions) of $L$, and $\{v^{0}_{n}\}_{n\geq1}$ the entries in the main diagonal of $L$. In the lower triangular part, for any $k>0$, $\{w^{-k}_n\}_{n\geq 1}$ and $\{w^{k}_n\}_{n\geq 1}$ occupy, respectively, odd and even positions on the $(2 k - 1)$-th diagonal. Similarly, the variables $\{v^{-k}_n\}_{n\geq 1}$ and $\{v^{k}_n\}_{n\geq 1}$ occupy odd and even positions on the $2k$-th diagonal. The evolution equations for the variables $v^{k}_{n}$ and $w^{k}_{n}$ follow from the Lax equation \eqref{eq:hamiltonian_L}.
For instance, 
the $t_1$-flow for the variables $v^{\pm k}_n$ and $w^{\pm k}_n$, respectively, is given by following equations
\begin{gather}
    \begin{aligned}
        \partial_{t_1} v^{k}_n =&\,\frac{1}{2} \left(v^0_{n-1} + v^0_n - v^0_{n-k -1} - v^0_{-k+n}\right)v^{k}_n + w^{k-1}_{n} - w^0_n w^{-(k+1)}_{n+1} \\[1.5ex]
    & - w^{-1}_n w^{-k}_n - w^0_{n-1}w^{-(k-1)}_{n-1}, \hspace{5ex} k<-1 \\[1.5ex]
    \partial_{t_1} v^{-1}_n =&\, \frac{1}{2}\left(v^0_{n-1} - v^0_{n+1}\right)v^{-1}_n + w^{-2}_n - w^0_n - w^{-1}_n w^1_n - w^0_{n-1}w^2_{n-1} \\[1.5ex]
    \partial_{t_1} v^0_n =&\, w^0_n w^1_n\\[1.5ex]
    \partial_{t_1} v^1_n =&\, \frac{1}{2}\left( v^0_{n+1} - v^0_{n-1}\right)v^1_n - w^{-2}_n + w^0_n + w^{-1}_{n+1}w^1_n+w^0_{n+1}w^2_n  \\[1.5ex]
    \partial_{t_1}v^k_n =&\, \frac{1}{2}\left(v^0_{k+n} + v^0_{k+n-1}-v^0_n - v^0_{n-1}\right)v^k_n + w^0_{n+k-1} w^{k-1}_n + w^{-1}_{n+k} w^k_n \\[1.5ex] & + w^0_{n+k} w^{k+1}_n - w^{-(k+1)}_n, \;\;\;\; k>1\\[1.5ex]
    \label{eq:t1v}
\end{aligned}
\end{gather}
and
\begin{gather}
\begin{aligned}
    \partial_{t_1} w^{k}_n =&\,\frac{1}{2} \left(v^0_{n-k-1} + v^0_{n-k-2}+ v^0_{n} + v^0_{n-1}\right)w^{k}_n + w^{0}_{n-k-2}v^{k+2}_{n} - w^0_n v^{-(k+2)}_{n+1} \\[1.5ex]
    &+ w^{-1}_{n-k-1} v^{k+1}_n - w^{-1}_{n}v^{-(k+1)}_{n} +w^0_{n-k-1} v^{k}_n - w^0_{n-1}v^{-k}_{n-1},\hspace{5ex} k<-1 \\[1.5ex]
    \partial_{t_1} w^{-1}_n =&  \,w^{0}_n v^{-1}_n - w^0_{n-1}v^1_{n-1}  \\[1.5ex]
    \partial_{t_1} w^0_n =& \,\frac{1}{2}\left(v^0_{n+1} -2 v^0_n + v^0_{n-1}\right)w^0_n\\[1.5ex]
    \partial_{t_1} w^k_n =&\, -\frac{1}{2}\left( v^0_{n+k} +  v^0_{n+k-1} +v^0_n + v^0_{n-1}\right)w^k_n +v^k_n - v^{-k}_n, \hspace{5ex} k > 0. \\[1.5ex]
\label{eq:t1w}
\end{aligned}
\end{gather}

\vspace{-10pt}
To give an idea of the increasing complexity of higher flows, the $t_2$-flows for both variables $v^{\pm k}_n$ and $w^{\pm k}_n$ are also reported in Appendix \ref{app:t2flow_all}.

We now consider the initial condition for the Lax matrix $L$. From 
\eqref{eq:l_t}, we have that the component~$w^{0}_n(\mathbf{t})$ can be expressed in terms of $\tau_{2n}(\mathbf{t})$ as follows (see \cite{vanMoerbekenotes})
\begin{equation}
w^{0}_n(\mathbf{t})=\bigg(\frac{\tau_{2n+2}(\mathbf{t})\,\tau_{2n-2}(\mathbf{t})}{\tau_{2n}^2(\mathbf{t})}\bigg)^{1/2}\,.\label{eq:initialw}
\end{equation} 
The function $\tau_{2n}(\mathbf{0})$ is given by a Selberg's integral which can be evaluated explicitly so that
\begin{equation}
\tau_{2n}(\mathbf{0})=\pi^{n/2}\, \prod_{k=0}^{n-1} 2^{-2k} (2k)!.
\label{eq:tau_tempi_zero}
\end{equation} 
Therefore, equations \eqref{eq:initialw} and \eqref{eq:tau_tempi_zero} imply
\begin{equation}
 w^{0}_n(\mathbf{0}) =  2 \sqrt{\pi} \sqrt{2n(2n-1)}.
\end{equation}
Similarly, using the expression for $v^0_n(\mathbf{t})$ obtained in \cite{vanMoerbekenotes}, i.e.
\begin{equation}
    v^{0}_n(\mathbf{t}) = \partial_{t_1}\log \tau_{2n}(\mathbf{t}).
\end{equation}
one can evaluate the initial datum  $v^0_n(\mathbf{0})$.
Hence, from the definition of  $\tau_{2n}(\mathbf{t})$ given in \eqref{eq:partition_and_tau_intro}, and due to the skew symmetry of the integration measure we have
\begin{equation}
v^{0}_n(\mathbf{0}) = 0.
\end{equation}
In general, the variables $v^{\pm k}_n$ are represented as suitable combinations of integrals of odd functions and therefore $v^{\pm k}_n(\mathbf{0}) = 0$.
We conlcude that the Lax matrix $L$ evaluated at $\mathbf{t}= \mathbf{0}$ takes the following form

\begin{equation}
\label{evenLax}
L({\bf{0}}) =\begin{pmatrix}
\hspace{2ex} 0 & 1 & 0 & 0 & 0 & 0 & 0 & 0 & \dots \hspace{2ex} \\[1.5ex]
\hspace{2ex}w_1^{-1}({\bf 0}) & 0 & w_1^0({\bf 0}) & 0 & 0 & 0 & 0 & 0 & \dots \hspace{2ex} \\[1.5ex]
\hspace{2ex} 0 & w_1^{1}({\bf 0}) & 0 & 1 & 0 & 0 & 0 & 0 & \dots \hspace{2ex} \\[1.5ex]
\hspace{2ex} w_1^{-2}({\bf 0}) & 0 & w_2^{-1}({\bf 0}) & 0 &  w_2^0({\bf 0})& 0 & 0 & 0 & \dots \hspace{2ex} \\[1.5ex]
\hspace{2ex} 0 & w_1^2({\bf 0}) & 0 & w_2^1({\bf 0}) & 0 & 1 & 0 & 0 & \dots \hspace{2ex} \\[1.5ex]
\hspace{2ex} w_1^{-3}({\bf 0}) & 0 & w_2^{-2}({\bf 0}) & 0 & w_3^{-1}({\bf 0}) & 0 & w_3^0({\bf 0}) & 0 & \dots \hspace{2ex} \\[1.5ex]
\hspace{2ex} 0 & w_1^3({\bf 0}) & 0 & w_2^{2}({\bf 0}) & 0 & w_3^1({\bf 0}) & 0 & 1 & \dots \hspace{2ex} \\[1.5ex]
\hspace{2ex} \vdots & \ddots & \ddots & \ddots & \ddots & \ddots & \ddots & \ddots & \ddots \hspace{2ex} \\[1.5ex]
\end{pmatrix} \\[1ex]
\end{equation}

\section{The reduced even Pfaff Hierarchy}
\label{sec:paffeven}
In this section we consider the SME with even power interactions specified by the partition function~\eqref{eq:partition_n_even_intro} and show that it provides a solution to a reduction of the even Pfaff Hierarchy i.e. the commuting flows \eqref{eq:hamiltonian_L} associated to the even times $t_{2k}$ only. 

In this case, equation \eqref{eq:tau_pfaffian}, i.e.  $\tau_{2n}(\mathbf{t}) = {\rm pf}(m_{2n}(\mathbf{t}))$, still holds with $m_{2n} = (\mu_{ij})_{0\leq i,j\leq 2n-1}$ and
\begin{equation}
\mu_{ij}(\mathbf{t}) = \langle x^i\,,\,y^j \rangle_{\mathbf{t}} = \int \int_{\mathbb{R}^2} x^i \, y^j \,\sigma(x-y)\,\text{e}^{\sum_{k\geq 1}t_{2k}\,(x^{2k}+y^{2k})}\text{e}^{-\frac{1}{2}(x^2 + y^2)}\,dx\,dy.
\label{eq:moments_even}
\end{equation}
Hence, the moments matrix $m_{2n}(\mathbf{t})$ reads as
\begin{equation}
m_{2n}=\big( \mu_{i\,j} \big)_{0\le i,j\le 2n-1}= \begin{pmatrix}
0 & \mu_{0\,1} & 0 & \mu_{0\,3} & 0 & \mu_{0\,5} & \dots \\[1ex]  
-\mu_{0\,1} & 0 & \mu_{1\,2} & 0 & \mu_{1\,4} & 0 & \dots \\[1ex]  
0 & -\mu_{1\,2} & 0 & \mu_{2\,3} & 0 & \mu_{2\,5} & \dots\\[1ex]  
-\mu_{0\,3} & 0 & -\mu_{2\,3} & 0 & \mu_{3\,4} & 0 & \dots\\[1ex]  
0 & -\mu_{1\,4} & 0 & -\mu_{3\,4} & 0 & \mu_{4\,5} & \dots\\[1ex]  
-\mu_{0\,5} & 0 & -\mu_{2\,5} & 0 & -\mu_{4\,5} & 0 & \dots \\[1ex] 
\vdots & \vdots & \vdots & \vdots & \vdots & \vdots & \ddots \\[1ex]  
\end{pmatrix}.
\end{equation} 
The moments \eqref{eq:moments_even} satisfy the evolution equations
\begin{equation}
\frac{\partial \mu_{ij}}{\partial t_{2k}}=\mu_{i+2k,j}+\mu_{i,j+2k}\end{equation}
which imply
\begin{equation}
\frac{\partial m_{\infty}}{\partial t_{2k}} = \Lambda^{2k}\,m_{\infty}+m_{\infty} \, \Lambda^{2k} \,.
\end{equation} 
We consider the reduction of the Lax equation \eqref{eq:hamiltonian_L} of the form
\begin{equation}
\frac{\partial L}{\partial t_{2k}} = \left[ \,- (L^{2k})_{\mathfrak{t}}\,, L \, \right]  \,,
\label{eq:hamiltonian_red_L}
\end{equation}
with 
\begin{equation}
L = \begin{pmatrix}
\hspace{2ex} 0 & 1 & 0 & 0 & 0 & 0 & 0 & 0 & 0 & \dots \hspace{2ex} \\[1.5ex]
\hspace{2ex}w_1^{-1} & 0 & w_1^0 & 0 & 0 & 0 & 0 & 0 & 0 & \dots \hspace{2ex} \\[1.5ex]
\hspace{2ex} 0 & w_1^{1} & 0 & 1 & 0 & 0 & 0 & 0 & 0 & \dots \hspace{2ex} \\[1.5ex]
\hspace{2ex} w_1^{-2} & 0 & w_2^{-1} & 0 &  w_2^0& 0 & 0 & 0 & 0 & \dots \hspace{2ex} \\[1.5ex]
\hspace{2ex} 0 & w_1^2 & 0 & w_2^1 & 0 & 1 & 0 & 0 & 0 & \dots \hspace{2ex} \\[1.5ex]
\hspace{2ex} w_1^{-3} & 0 & w_2^{-2} & 0 & w_3^{-1} & 0 & w_3^0 & 0 & 0 & \dots \hspace{2ex} \\[1.5ex]
\hspace{2ex} 0 & w_1^3 & 0 & w_2^{2} & 0 & w_3^1 & 0 & 1 & 0 & \dots \hspace{2ex} \\[1.5ex]
\hspace{2ex} \vdots & \ddots & \ddots & \ddots & \ddots & \ddots & \ddots & \ddots & \ddots & \ddots \hspace{2ex} \\[1.5ex]
\end{pmatrix},\label{eq:L_even} \\[1ex] 
\end{equation}
that is the Lax matrix associated to the SME with even power interactions is obtained from the general one by setting the variables $v^0_n$, $v^{\pm k}_n$ identically equal to zero for any $\mathbf{t}_{k}$. In other words, the partition function gives a solution to a reduction of the even Pfaff Hierarchy which preserves the zeros of the initial Lax matrix $L(\mathbf{0})$ given by the expression \eqref{evenLax}. 

The first non-trivial flow of the reduced even Lax hierarchy \eqref{eq:hamiltonian_red_L} provides the following evolution equations for the variables $w^k_n$

\begin{gather}
\label{pfaff_lattice}
\begin{aligned}
\partial_{t_2} w^{k}_{n} =&~ \frac{1}{2} \left(w^{k}_{n} w^{0}_{n} w^{1}_{n} + w^{k}_{n} w^{0}_{n-k-1} w^{1}_{n-k-1} - w^{k}_{n} w^{0}_{n-1} w^{1}_{n-1} - w^{k}_{n} w^{0}_{n-k-2} w^{1}_{n-k-2} \right) \\[1.2ex]
&+w^{k+1}_{n+1} w^{0}_{n} + w^{k-1}_{n} w^{0}_{n-k-1} - w^{k-1}_{n-1} w^{0}_{n-1} - w^{k+1}_{n} w^{0}_{n-k-2}, \qquad k <- 1\\[1.2ex]
\partial_{t_2} w^{-1}_{n}=&~ w^{0}_{n} \left( w^{-1}_{n} w^{1}_{n} + w^{-2}_{n} + w^{0}_{n} \right) - w^{0}_{n-1} \left(w^{-1}_{n} w^{1}_{n-1} + w^{-2}_{n-1} \right) - \left( w^{0}_{n-1} \right)^{2} \\[1.2ex]
\partial_{t_2} w_{n}^{0} =&~ \frac{1}{2} \left( w^{0}_{n+1} w^{1}_{n+1} - w^{0}_{n-1} w^{1}_{n-1} \right)w^0_n + \left(w^{-1}_{n+1} - w^{-1}_{n}\right)w^0_n  \\[1.2ex]
\partial_{t_2} w_{n}^{1} =&~ \frac{1}{2} \left( w^{0}_{n-1} w^{1}_{n-1} w^{1}_{n}  - w^{0}_{n+1} w^{1}_{n} w^{1}_{n+1}\right) + w^{0}_{n+1} w^{2}_{n} - w^{0}_{n-1} w^{2}_{n-1} \\[1.2ex]
\partial_{t_2} w^{k}_{n} =&~ \frac{1}{2} \left(w^{0}_{n-1} w^{1}_{n-1} w^{k}_{n} + w^{0}_{n+k-1} w^{1}_{n+k-1} w^{k}_{n} -w^{0}_{n} w^{1}_{n} w^{k}_{n} - w^{0}_{n+k} w^{1}_{n+k} w^{k}_{n} \right) \\[1.2ex]
&+ w^{0}_{n} w^{k-1}_{n+1} + w^{0}_{n+k} w^{k+1}_{n} - w^{0}_{n-1} w^{k+1}_{n-1} - w^{0}_{n+k-1} w^{k-1}_{n}, \qquad k> 1.
\end{aligned}
\end{gather}
The derivation described above naturally compares with the case of the HME, as studied in \cite{HermitianPRE}, where the partition function corresponds to a particular solution of the Toda Lattice and the reduction to even power interactions provides a solution to the Volterra Lattice. The Volterra Lattice is effectively an independent integrable system as it arises from a reduction of the even Toda Hierarchy and it is not compatible with the odd flows of the Toda Hierarchy. Similarly, the reduction of the even Pfaff Hierarchy obtained from the SME with even interactions is not compatible with the odd flows of the Pfaff Hierarchy.

\section{Thermodynamic limit and integrable hydrodynamic chain}
\label{sec:limit}
We study the large $n$ asymptotic properties of the SME via the continuum limit of the Pfaff Lattice equations. In particular, we focus on the case of even power interactions \eqref{eq:partition_n_even_intro} described by the Lax matrix \eqref{eq:L_even} which satisfies the Lax equations \eqref{eq:hamiltonian_red_L}. 

As observed above, the lattice equations for the reduced even Pfaff Hierarchy \eqref{pfaff_lattice} constitute an infinite chain for the variables $w^{k}_{n}$, where $k \in \mathbb{Z}$ labels the components of the chain and $n \in \mathbb{N}$ labels points on the lattice. In Section \ref{sec:pfaffgen}, we noted that the variables $w^{k}_{n}$ can be expressed in terms of suitable elements of the sequence of $\tau-$functions $\{\tau_{2n}\}_{n\geq 1}$ and their derivatives. As~$n \to \infty$, for the variables $w^{k}_{n}$ we have $w^{k}_{n+1} - w^{k}_{n} = O(\varepsilon)$, with $\varepsilon \to 0$ such that  $x = \varepsilon n$ remains finite. In the following, we derive the continuum limit equations for the chain and study the integrability at the leading order with respect to the $\varepsilon$ expansion. We illustrate the result for the first equation of the hierarchy given by the $t_{2}$-flows. As mentioned in section \ref{sec:intro}, our considerations extend to the $t_4$- and $t_6$-flows as well, and we conjecture they hold for any equation of the hierarchy.

Let us introduce the interpolation function $w^{k}(x/\varepsilon)$ with $x = \varepsilon n$ 
so that $w^{k}(n) = w^{k}_{n}$, and define
$$
u^{k}(x) := w^{k}\left (\frac{x}{\varepsilon} \right)
$$ 
 with $u^{k}(x \pm \varepsilon ) = w^{k}_{n \pm 1}$. Substituting $u^k(x)$ into the equations~\eqref{pfaff_lattice}, expanding in Taylor series  for $\varepsilon \to 0$ and setting $t = \varepsilon \, t_{2}$, at the leading order $O(\varepsilon^{0})$ we get the following system of PDEs
\begin{gather}
\label{hydrochain}
\begin{aligned}
u^{k}_{t} =&\left( (k+2) u^{k+1} - k u^{k-1} + u^{1} u^{k} \right) u^{0}_{x}  + u^{0} u^{k} u^{1}_{x}  + u^{0} u^{k-1}_{x} + u^{0} u^{k+1}_{x}, \qquad k < 0\\[1.2ex]
u^{0}_{t} =&~u^{0} u^{1} u^{0}_{x} + \left (u^{0} \right)^{2} u^{1}_{x} + u^{0} u^{-1}_{x} \\[1.2ex]
u^{1}_{t} =& \left (2 u^{2} - \left (u^{1} \right )^{2} \right ) u^{0}_{x} - u^{0} u^{1} u^{1}_{x} + u^{0} u^{2}_{x} \\[1.2ex]
u^{k}_{t} =& \left( (k+1) u^{k+1} - (k-1) u^{k-1} - u^{1} u^{k} \right) u^{0}_{x} - u^{0} u^{k} u^{1}_{x} + u^{0} u^{k-1}_{x} + u^{0} u^{k+1}_{x}, \qquad k>1
\end{aligned}
\end{gather} 
with the notation $f_{t} = \partial_{t} f$, $f_{x} = \partial_{x} f$. In particular, we note that the system~\eqref{hydrochain} is an infinite chain of quasilinear PDEs of hydrodynamic type. In fact, the equations of the chain are of the form
\begin{equation}
u^{k}_{t} =  a^{k}_{0} \, u^{0}_{x} +a^{k}_{1} \,  u^{1}_{x} +a^{k}_{k-1} \, u^{k-1}_{x} + a^{k}_{k+1}\,  u^{k+1}_{x}
\end{equation}
or equivalently
\begin{equation}
\label{hydromatrixchain}
{\bf u}_{t} = A({\bf u}) {\bf u}_{x} 
\end{equation}
where $A ({\bf u}) = \{ a^{k}_{j} \}_{j,k=-\infty}^{+\infty}$ is an infinite matrix such that $a^{k}_{j} = 0$ if $\notin \{0,1,k-1,k+1 \}$ and
\begin{gather}
\begin{aligned}
a^{k}_{0}&=   \begin{cases}
\hspace{1ex} (k+2) u^{k+1} - k u^{k-1} + u^{1} u^{k} &\textup{if} \quad k<0  \\[1.2ex]
\hspace{1ex} u^{0} u^{1} &\textup{if} \quad k = 0  \\[1.2ex]
\hspace{1ex} (k+1) u^{k+1} - (k-1) u^{k-1} - u^{1} u^{k} &\textup{if} \quad k \geq 1  \\[1.2ex]
\end{cases} \hspace{6ex} a^{k}_{1} =  \begin{cases}
\hspace{1ex} u^{0} u^{k} &\textup{if} \quad k \leq 0  \\[1.2ex]
\hspace{1ex}  - u^{0} u^{k} &\textup{if} \quad k \geq 1  
\end{cases}  \\[1.2ex] 
a^{k}_{k-1}&= \begin{cases}
\hspace{1ex} u^{0} &\textup{if} \quad k \neq  1  \\[1.2ex]
\hspace{1ex}\left(2 u^2 - (u^1)^2\right) &\textup{if} \quad k = 1
\end{cases} \hspace{6ex}  
 a^{k}_{k+1} =  \begin{cases}
\hspace{1ex} u^{0} &\textup{if} \quad k \neq 0  \\[1.2ex]
\hspace{1ex}  (u^0)^2 &\textup{if} \quad k = 0  
\end{cases}
\end{aligned}
\label{eq:a}
\end{gather}
By applying the same procedure, one can construct a hierarchy of infinitely many commuting flows, each of them in the form of a hydrodynamic chain \eqref{hydromatrixchain_intro} from the thermodynamic limit of the higher flows of the hierarchy \eqref{eq:hamiltonian_red_L}. 
The hydrodynamic chain \eqref{hydrochain} is integrable as it possesses an infinite hierarchy of commuting flows.
In the following, we show that the hydrodynamic chain \eqref{hydrochain} is diagonalisable and integrable according to  the criterion introduced in~\cite{chain}, namely the existence of integrable hydrodynamic reductions in an arbitrary number of components. 

Following \cite{chain}, the diagonalisability of the hydrodynamic chain is established by studying the Haantjes tensor
\begin{equation}
H^{i}_{jk} = N^{i}_{pr} \,  a^{p}_{j} \,  a^{r}_{k} - N^{p}_{jr} \,  a^{i}_{p} \,  a^{r}_{k} - N^{p}_{rk} \,  a^{i}_{p} \,  a^{r}_{j} + N^{p}_{jk} \, a^{i}_{r} \, a^{r}_{p}\label{eq:haantjes}
\end{equation}
where $N^i_{jk}$ is the Nijenhuis tensor
\begin{equation}
N^{i}_{jk} = a^{p}_{j} \,  \partial_{u^{p}} a^{i}_{k} -a^{p}_{k} \, \partial_{u^{p}} a^{i}_{j} - a^{i}_{p} \left(\partial_{u^{j}} a^{p}_{k} -\partial_{u^{k}} a^{p}_{j} \right).\label{eq:Ndef}
\end{equation}
In the case of infinite matrices, both Nijenhuis and Haantjes  tensors are well defined for the so called matrices of {\it chain class}.
\begin{definition}[Chain class matrices \cite{chain}] An infinite matrix $A({\bf u})$ is said to belong to the class $C$ (chain class) if it satisfies the following two properties:
\begin{enumerate}
\item[a)] each row of $A({\bf u})$ contains finitely many non-zero elements;
\item[b)] each matrix element of $A({\bf u})$ depends on finitely many variables $u^{k}$.
\end{enumerate}
\end{definition}
Bearing in mind the form of the matrix $A({\bf u})$ as specified in \eqref{eq:a}, we have the following
\begin{proposition}
Given the chain \eqref{hydrochain}, the associated matrix $A({\bf u})$ in \eqref{hydromatrixchain} belongs to the chain class.
\end{proposition}
Moreover, based on the Haantjes theorem given in \cite{HaantjesTH}, the following definition extends the concept of diagonalisability to the case of infinite matrices:
\begin{definition}[Diagonalisable hydrodynamic chains~\cite{chain}]
A hydrodynamic chain from the class $C$ is said to be diagonalisable if all components of the corresponding Haantjes tensor are zero.
\end{definition}
We show that our chain fulfils the definition above.
\begin{proposition}
\label{th:Htest}
Given the chain \eqref{hydrochain}, the Haantjes tensor of the associated matrix $A({\bf u})$  vanishes. \label{prop:diagonalisability}
\end{proposition}
\begin{proof} The proof proceeds by direct inspection. Observing that by definition $N^i_{jk}$ is antisymmetric under exchange of $j$ and $k$, i.e. $N^i_{jk} = - N^i_{kj}$, a direct calculation shows that $N^0_{jk} = 0$ for any $j$ and $k$. Similarly, for $i \neq 0$ the only nonzero elements of $N^i_{jk}$ are
\begin{equation}
N^i_{0\,\pm1}, \,N^{i}_{0\,i}, \,N^{i}_{0\,i\pm 1}, \,N^i_{1\,i \pm 1} \,N^{i}_{-1\,i\pm 1}, \,N^{i}_{-1\,+1}
\label{eq:N_nonzero}
\end{equation}  
and their counterparts with the lower indices exchanged. The above components can be computed for a generic value of $i$, and their explicit expressions are listed in Appendix \ref{sec:N}. 
The structure of $N^{i}_{jk}$ and~$A({\bf u})$ induces constraints on the range of values the indices $p$ and $r$ can take in the expression of the Haantjies tensor \eqref{eq:haantjes}, and consequently on potential nonzero elements. Indeed, the form of $N^i_{jk}$, specified by the elements \eqref{eq:N_nonzero},  implies that for any fixed $i$ the only components of $H^i_{jk}$ which are not trivially zero are those for $j, k\in\{0,\pm 1,\pm 2, 3, i, i\pm1, i\pm2, i\pm3\}$. Given the explicit expressions for $a^k_j$ in \eqref{eq:a} and $N^i_{jk}$ in Appendix \ref{sec:N}, a direct calculation demonstrates that $H^i_{jk} = 0$ for the listed values of the lower indices. This proves the statement. 
\end{proof}

We now study the integrability of the chain \eqref{hydromatrixchain} by following the approach based on the method of hydrodynamic reductions applied to the system \eqref{hydrochain}. We look for solutions of the form 
\begin{equation}
\label{Nphases}
u^{k} = u^{k}(R^{1},R^{2},\dots, R^{N})
\end{equation} 
for an arbitrary number $N$ of components $R^{i} = R^{i}(x,t)$. The functions $\{R^i\}_{i=1}^N$ are the Riemann invariants and satisfy by definition the diagonal system
\begin{equation}
\label{Rinvariant}
R^{i}_{t} = \lambda^{i}(R^{1},\dots,R^{N}) R^{i}_{x}
\end{equation}
where the characteristic speeds $\lambda^{i}$ are such that the system \eqref{Rinvariant} possesses the semi-Hamiltonian property, that is
\begin{equation}
\label{semiH}
\partial_{k} \left (\frac{\partial_{j} \lambda^{i}}{\lambda^{j} - \lambda^{i}} \right) = \partial_{j} \left (\frac{\partial_{k} \lambda^{i}}{\lambda^{k} - \lambda^{i}} \right),
\end{equation}
with the notation $\partial_{i} = \partial_{R^{i}}$.
The diagonal form of the system \eqref{Rinvariant} and the condition \eqref{semiH} guarantee that equations \eqref{Rinvariant} constitute a system of conservation laws \cite{Sevennec} which is integrable via the generalised hodograph method \cite{Tsarev}. Substituting the assumption~\eqref{Nphases} into the system~\eqref{hydromatrixchain} and using \eqref{Rinvariant} we obtain the equations of the form
\begin{equation}\
\label{eigenA}
\lambda^{i} \, \partial_{i} {\bf u} = A({\bf u}) \, \partial_{i} {\bf u}, \qquad i = 1,2,\dots N
\end{equation}
where we used the fact that $R^{i}_{x}$ for $i=1,\dots,N$ are independent. We observe that, due to the specific sparse structure of the matrix $A({\bf u})$,  the components of the eigenvectors $\partial_{i}{{\bf u}}$ can be parametrised  in terms of the components $\partial_{i} u^{0}$ and $\partial_{i} u^{1}$. 

Let us consider, for example, the equations for $\partial_{i} u^{-2}$, $\partial_{i} u^{-1}$, $\partial_{i} u^{2}$ and $\partial_{i} u^{3}$:
\begin{gather}
\label{hchain}
\begin{aligned} 
\partial_{i}{u^{-2}} =&\, \frac{1}{(u^{0})^{2}} \left ( (\lambda^{i})^{2} - u^{0} u^{1} \lambda^{i} - u^{0} (2  u^{0} + u^{-2} + u^{-1} u^{1}) \right) \partial_{i} u^{0} - \left( \lambda^{i} + u^{-1} \right) \partial_{i} u^{1} \\[1.2ex]
\partial_{i}{u^{-1}} =&\, \left(\frac{\lambda^{i}}{u^{0}} - u^{1} \right) \partial_{i} u^{0} - u^{0} \partial_{i} u^{1} \\[1.2ex]
\partial_{i}{u^{2}} =&\, \frac{1}{u^{0}} \left((u^{1})^{2} - 2 u^{2} \right) \partial_{i} u^{0} + \frac{1}{u^{0}} \left(\lambda^{i} + u^{0} u^{1}  \right) \partial_{i} u^{1} \\[1.2ex] 
\partial_{i}{u^{3}} =&\, \frac{1}{(u^{0})^{2}} \left( \left ((u^{1})^{2} - 2 u^{2} \right ) \lambda^{i} + u^{0} \left(u^{1} (1 + u^{2}) - 3 u^{3} \right) \right) \partial_{i} u^{0} \nn \\ &+ \frac{1}{(u^{0})^{2}} \left ( (\lambda^{i})^{2} + u^{0} u^{1} \lambda^{i} + (u^{0})^{2} (u^{2} -1) \right) \partial_{i} u^{1}\,, \qquad i = 1,\dots, N.
\end{aligned}
\end{gather}
The compatibility conditions
\[
\partial_{j}\partial_{i} u^{-2} = \partial_{i}\partial_{j} u^{-2} \qquad  \partial_{j}\partial_{i} u^{-1} = \partial_{i}\partial_{j} u^{-1}  \qquad \partial_{j}\partial_{i} u^{2} = \partial_{i}\partial_{j} u^{2} \qquad \partial_{j}\partial_{i} u^{3} = \partial_{i}\partial_{j} u^{3}
\]
lead to a so called Gibbons-Tsarev system. For our chain this takes the form
\begin{gather} 
\label{GT}
\begin{aligned} 
\partial_{j} \lambda^{i} &= \frac{4 (u^{0})^{2} - \lambda^{i} \lambda^{j}}{u^{0} (\lambda^{i} - \lambda^{j})} \partial_{j} u^{0} \\[1.2ex]
\partial_{i} \lambda^{j} &= \frac{4 (u^{0})^{2} - \lambda^{i} \lambda^{j}}{u^{0} (\lambda^{j} - \lambda^{i})} \partial_{i} u^{0}  \\[1.2ex]
\partial_{i} \partial_{j} u^{0} &= \frac{(\lambda^{i})^{2} + (\lambda^{j})^{2} - 8 (u^{0})^{2}}{u^{0} (\lambda^{i} - \lambda^{j})^{2}} \partial_{i} u^{0} \partial_{j} u^{0}  \\[1.2ex]
\partial_{i} \partial_{j} u^{1} &= - \frac{(\lambda^{j} - 2 \lambda^{i}) \lambda^{j} + 4 (u^{0})^{2}}{u^{0} (\lambda^{i} - \lambda^{j})^{2}} \partial_{i} u^{0} \partial_{j} u^{1} - \frac{(\lambda^{i} - 2 \lambda^{j}) \lambda^{i} + 4 (u^{0})^{2}}{u^{0} (\lambda^{i} - \lambda^{j})^{2}} \partial_{j} u^{0} \partial_{i} u^{1}.
\end{aligned}
\end{gather}  
A direct calculation shows that the system of equations \eqref{GT} is in involution, i.e. compatibility conditions of the form
\[
\partial_{k} \partial_{j} \lambda^{i} = \partial_{j} \partial_{k} \lambda^{i} \qquad  \partial_{k} \partial_{i} \partial_{j} u^{0} =\partial_{i} \partial_{k} \partial_{j} u^{0} \qquad  \partial_{k} \partial_{i} \partial_{j} u^{1} =\partial_{i} \partial_{k} \partial_{j} u^{1}
\]
are satisfied modulo the equations~\eqref{GT} for all permutation of the derivatives with respect to $R^{i}$, $R^{j}$, $R^{k}$. A first classification of Gibbons-Tsarev systems has been provided by Odesskii and Sokolov \cite{ChainsOdesskiSokolov, HydroOdesskiSokolov}. We note that, at the best of our knowledge, the system \eqref{GT} has not appeared before in the literature and it is not included in the class considered in \cite{ChainsOdesskiSokolov, HydroOdesskiSokolov}.

The compatibility of the Gibbons-Tsarev system \eqref{GT} guarantees that for any solution of the Riemann invariants system \eqref{hydrochain} it is possible to construct  a solution of the hydrodynamic chain. This property was proposed in \cite{chain} as definition of integrability of a hydrodynamic chain:
\begin{definition}[Integrable hydrodynamic chains~\cite{chain}]
A hydrodynamic chain of class $C$ is integrable if it admits $N-$phase solutions of the form \eqref{Nphases} for arbitrary $N$.
\end{definition}

Therefore, the above calculations prove the following 
\begin{theorem}
\label{th:hrydrored}
The hydrodynamic chain \eqref{hydrochain} is integrable in the sense of the hydrodynamic reductions.
\end{theorem}

\section{Concluding remarks}
\label{sec:conclusions}
Extensive studies of Random Matrix Ensembles and their connection with the theory of integrable systems (see e.g. \cite{vanMoerbekenotes} and references therein) show that the order parameters, defined as derivatives of the partition function, and their suitable combinations appear as entries of the Lax matrix and the associated Lax equation. For example, in the case of HME, one has the Lax equations for the Toda Lattice. The matrix ensemble of interest is specified by a particular solution of the hierarchy, obtained from a suitable initial condition. Such initial condition is evaluated by considering the partition function and its derivatives in the case where all coupling constants $t_k$ vanish.
Similarly, in the case of SMEs the underlying integrable system is constituted by the Pfaff Lattice and the equations of its hierarchy. These equations specify the behaviour of the order parameters, namely the entries of the Lax matrix, as functions of the coupling constants. 

For even power interactions, the thermodynamic limit of the HME is constituted by an order parameter that evolves according to a scalar integrable hierarchy (the Hopf hierarchy) \cite{HermitianPRE}. On the other hand, in the $n\to \infty$ limit, the even SME  is specified  by infinitely many order parameters that satisfy an integrable hydrodynamic chain. 
This result follows from the key observation that the components of the reduced even Pfaff Lattice can be rearranged in the form of a chain of equations, where the state of each site is given by a vector of infinitely many components.
From this point of view the SME reveals a higher level of complexity compared to the HME due to the existence of integrable reductions in any number of components and associated critical scenarios. 

It is indeed well known that, for generic initial conditions, solutions of systems of hydrodynamic type break down in finite time, namely, in the context of SME, for finite values of the coupling constants $t_{k}$. In the case of the HME, the critical behaviour of the order parameter at the leading order is described by the Whitney cusp and, as observed in \cite{HermitianPRE}, finite size corrections resolve the singularity via the onset of a modulated highly oscillating quasi-periodic wave, known as  dispersive shock. The dispersive shock characterises a new type of phase transition where asymptotic stable states are connected by an intermediate state where order parameters develop fast oscillations induced by the dispersive nature of finite size corrections. The case of SME presents a similar scenario with a potentially richer variety of behaviours due to the higher number of components. Further studies in this direction will entail the detailed analysis of the solution with the specific initial condition induced by the partition function \eqref{eq:partition_n_even_intro} calculated at $\mathbf{t} =\mathbf{0}$. 

The application of the Haantjes tensor test and the method of hydrodynamic reductions allow to prove the integrability of the hydrodynamic chain by considering the first nontrivial flow only. Integrability implies the existence of infinitely many commuting flows which describe the evolution of the order parameters in the space of coupling constants.

We finally note that above considerations are concerned with a direct comparison between HME and SME when restricted to even power interactions. It is important to note that, with the given scaling, the hydrodynamic chain arises in the case of even power interactions only. As  mentioned earlier and further specified in Appendix  \ref{app:dispfaff}, the first flow \eqref{t1flows} associated to $t_{1}$ does not lead to an infinite chain of quasilinear PDEs.
The system \eqref{t1flows} and its relation with the large $n$ scaling properties of the initial condition will be analysed in detail in a separate work.
The previously unseen connection between matrix ensembles and hierarchies of hydrodynamic chains discussed in this paper, together with the aforementioned results for the HME discussed in \cite{HermitianPRE} 
suggests that the study of random matrix models may lead to the discovery of new interesting integrable hydrodynamic PDEs. The study of the PDEs so obtained arises as a general framework and a new methodology to classify and describe asymptotic properties of complex systems.

\newpage
\bibliography{biblio}
\bibliographystyle{spmpsc}
\newpage
\appendix
\section{$t_2$-flow of the Pfaff Lattice}
\label{app:t2flow_all}
We provide the expressions for the second flow of the Pfaff Lattice with both even and odd times. By direct inspection the evolution equations for $v^k_n$ and $w^k_n$ with respect to $t_2$ read respectively
\begin{gather} 
\label{pfaff_lattice_fields_v}
\begin{aligned} 
\partial_{t_{2}} v^{k}_n & = -\frac{1}{2}\,v^{k}_n \left( (v^0_{n-k})^2-(v^0_{n-k-1})^2- (v^0_n)^2+(v^0_{n-1})^2   + w^0_{n-k}\,w^1_{n-k}  -w^0_{n-k-1}\,w^1_{n-k-1}\right. \\[1.5ex]
& \left.\hspace{2ex} - w^0_n\,w^1_n+w^0_{n-1}\,w^1_{n-1} \right)    + w^0_{n-k}\,v^{k-1}_n - w^0_{n-k-1}\,v^{k+1}_n  + w^0_n\,v^{k+1}_{n+1} \\[1.5ex]
& \hspace{2ex} - w^0_{n-1}\,v^{k-1}_{n-1}  + \left( v^0_{n-k}-v^0_{n-k-1} \right)\,w^{k-1}_{n} - \left( v^0_n-v^0_{n-1} \right)\,w^{-1}_n\,w^{-k}_{n} \\[1.5ex]
& \hspace{2ex} - \left( w^0_n\,v^{k+1}_n - w^0_{n-1}\,v^{-(k+1)}_{n-1} \right)\,w^{-k}_{n}, \, \hspace{10ex} k < 0\\[1.5ex] 
\partial_{t_{2}} v^0_n &= w^0_n\left( v^1_n + v^{-1}_n\right) \\[1.5ex]
\partial_{t_{2}} v^{k}_n & = \frac{1}{2}v^{k}_n \left(( v^0_{n+k})^2-(v^0_{n+k-1})^2+ (v^0_n)^2-(v^0_{n-1})^2   + w^0_{n+k}\,w^1_{n+k}  -w^0_{n+k-1}\,w^1_{n+k-1}\right. \\[1.5ex]
& \hspace{2ex} \left.+ w^0_n\,w^1_n-w^0_{n-1}\,w^1_{n-1} \right)   + w^0_{n+k}\,v^{k+1}_n - w^0_{n+k-1}\,v^{k-1}_n  + w^0_n\,v^{k}_{n+1} - w^0_{n-1}\,v^{k+1}_{n-1} \\[1.5ex]
& \hspace{2ex} + \left( v^0_{n+k}-v^0_{n+k-1} \right)\,w^{-1}_{n+k}\,w^{k}_{n} - \left( v^0_n-v^0_{n-1} \right)\,w^{-(k+1)}_{n} \\[1.5ex]
& \hspace{2ex} + \left( w^0_{n+k}\,v^{-(k-1)}_{n+k} - w^0_{n+k-1}\,v^{k-1}_{n+k-1}  \right)\,w^{k}_{n}\,, \hspace{10ex} k > 0 \\[1.5ex]
\end{aligned}
\end{gather}
\begin{gather}
\label{pfaff_lattice_fields_w}
\begin{aligned} 
 \partial_{t_{2}} w^{k}_{n}  =& \frac{1}{2}w^{k}_{n} \left((v^0_{n-k-1})^2-(v^0_{n-k-2})^2+ (v^0_n)^2-(v^0_{n-1})^2\right. \\[1.5ex]
& \left. \hspace{2ex} + w^0_{n-k -1}\,w^1_{n-k-1}  -w^0_{n-k-2}\,w^1_{n-k-2}  + w^0_n\,w^1_n-w^0_{n-1}\,w^1_{n-1} \right) \\[1.5ex]
& \hspace{2ex}  + w^0_{n-k-1}\,w^{k-1}_{n} - w^0_{n-k-2}\,w^{k+1}_{n}  + w^0_n\,w^{k+1}_{n+1} \\[1.5ex]
& \hspace{2ex} - w^0_{n-1}\,w^{k-1}_{n-1}    + \left( v^0_{n-k-1}-v^0_{n-k-2} \right)\,w^{-1}_{n-k-1}\,v^{k+1}_n - \left( v^0_n-v^0_{n-1} \right)\,w^{-1}_n\,v^{-(k+1)}_n \\[1.5ex]
& \hspace{2ex}  + \left( w^0_{n-k-1}\,v^{k+2}_{n-k-1} + w^0_{n-k-2}\,v^{-(k+2)}_{n-k-2} \right)\,v^{k+1}_n \\[1.5ex]
& \hspace{2ex} - \left( w^0_n\,v^{k+2}_n + w^0_{n-1}\,v^{-(k+2)}_{n-1} \right)\,v^{-(k+1)}_n\,, \hspace{5ex}  k < 0 \\[1.5ex]
\partial_{t_{2}} w^0_n =& \frac{1}{2} w^0_n\,\left( (v^0_{n+1})^2-(v^0_{n-1})^2+w^0_{n+1}\,w^1_{n+1}-w^0_{n-1}\,w^1_{n-1} \right)  + w^0_n\left( w^{-1}_{n+1}-w^{-1}_{n-1} \right)  \\[1.5ex]
\partial_{t_{2}} w^{1}_n  =& -\frac{1}{2}w^{1}_n \left(( v^0_{n+1})^2-(v^0_{n-1})^2   + w^0_{n+1}\,w^1_{n+1}-w^0_{n-1}\,w^1_{n-1} \right) + w^0_{n+1}\,w^{2}_n   - (w^0_{n})^2  \\[1.5ex]
& \hspace{2ex} + w^0_n\,w^{0}_{n+1}- w^0_{n-1}\,w^{2}_{n-1}   + \left( v^0_{n+1}-v^0_{n} \right)\,v^{1}_n  - \left( v^0_n-v^0_{n-1} \right)\,v^{-1}_n \\[1.5ex]
\partial_{t_{2}} w^{k}_n =& -\frac{1}{2}w^{k}_n \left(( v^0_{n+k})^2-(v^0_{n+k-1})^2+ (v^0_n)^2-(v^0_{n-1})^2   + w^0_{n+k}\,w^1_{n+k}\right.\\[1.5ex]
& \hspace{6ex}\left. -w^0_{n+k-1}\,w^1_{n+k-1}    +  w^0_n\,w^1_n-w^0_{n-1}\,w^1_{n-1} \right) + w^0_{n+k}\,w^{k+1}_n   - w^0_{n+k-1}\,w^{k-1}_{n}  \\[1.5ex]
& \hspace{2ex} + w^0_n\,w^{k-1}_{n+1}- w^0_{n-1}\,w^{k+1}_{n-1}   + \left( v^0_{n+k}-v^0_{n+k-1} \right)\,v^{k}_n  - \left( v^0_n-v^0_{n-1} \right)\,v^{-k}_n\,, \hspace{5ex}  k > 1. \\[1.5ex]
\end{aligned} 
\end{gather}

\newpage
\section{Continuum limit of the Pfaff Lattice: \texorpdfstring{$t_{1}-$}{t1}flow and higher order corrections}
\label{app:dispfaff}
We provide the continuum limit of equations \eqref{eq:t1v}, \eqref{eq:t1w} constituting the $t_1$-flow of the Pfaff Lattice. Using the same approach described in Section \ref{sec:limit} for the reduced even Pfaff Lattice, we introduce interpolating functions $w^{k}(x/\varepsilon)$ and $v^{k}(x/\varepsilon)$, with finite $x = \varepsilon n$, in the limit $n\rightarrow \infty$ and $\varepsilon \to 0$. Hence, given $w^{k}(n) = w^{k}_{n}$, $v^k(n) = v^k_n$ we define
\begin{equation*}
\begin{split} 
& u^{k}(x) := w^{k}\left (\frac{x}{\varepsilon} \right) \\[1.5ex]
& z^k(x) := v^k\left(\frac{x}{\varepsilon}\right).
\end{split}
\end{equation*}
Therefore, at $x = \varepsilon n$, we have $u^{k}(x \pm \varepsilon ) = w^{k}_{n \pm 1}$ and $z^{k}(x \pm \varepsilon ) = v^{k}_{n \pm 1}$. Using the above substitution into the equations \eqref{eq:t1v} and \eqref{eq:t1w} and expanding in Taylor series  for $\varepsilon \to 0$, the evolution equations for $v^k_n$ \eqref{eq:t1v} give, up to  $O(\varepsilon^3)$ 
\begin{gather}
\label{t1flows}
\begin{aligned}
    z^{k}_{t_1} =&  - u^0 \left(u^{-(k+1)}+u^{-(k-1)}\right)-u^{-1}u^{-k} +u^{k-1}\\[1.5ex]
    &\hspace{2ex}+ \varepsilon\left( u^{-(k-1)}u^0_x  + k\,z^{k}z^0_x- u^0 u^{-(k+1)}_x +  u^0 u^{-(k-1)}_x\right)\\[1.5ex] &\hspace{2ex}-\frac{\varepsilon^2}{2}\left((u^{-(k-1)}u^0)_{xx} + u^0 u^{-(k+1)}_{xx}+\,(k+1)k\,z^{k} z^0_{xx}\right)\\[1.5ex]
    &\hspace{2ex}+ O(\varepsilon^3),\hspace{10ex}k<-1\\[1.5ex]
     z^{-1}_{t_1} =& \,u^{-2} - u^0 - u^{-1} u^1 -u^0 u^2 + \varepsilon\left(u^0_x u^2 + u^0u^2_x -z^{-1} z^0_x\right) -\frac{\varepsilon^2}{2} \left(u^0 u^2\right)_{xx} + O(\varepsilon^3)\\[1.5ex]
    z^0_{t_1} =& \,u^0 u^1\\[1.5ex]
    z^1_{t_1}=& - u^{-2} + u^0 + u^{-1} u^1 + u^0 u^2 + \varepsilon\left(z^0_x z^1 + u^{-1}_x u^1 + u^0_x u^2\right)\\[1.5 ex]
    &+ \frac{\varepsilon^2}{2}\left(u^1 u^{-1}_{xx} + u^{2}u^0_{xx}\right) + O(\varepsilon^3) \\[1.5ex]
    z^k_{t_1} =& \, u^0\left(u^{k-1} + u^{k+1}\right) + u^{-1} u^k - u^{-(k+1)} \\[1.5ex]
    &+ \varepsilon\left(k\, z^0_x z^k + (k-1) u^{k-1} u^0_x + k u^{-1}_x u^k + k u^0_x u^{k+1}\right)\\[1.5ex]
    &+\frac{\varepsilon^2}{2}\left(k^2 (u^k u^{-1}_{xx}+u^{k+1} u^0_{xx}) + (k-1)^2 u^{k-1} u^0_{xx} + k(k-1) z^k z^0_{xx} \right)\\[1.5ex] 
    &+ O(\varepsilon^3),\hspace{10ex} k>1.\\[1.5ex]
    \end{aligned}
    \end{gather}
Observe that at the leading order $O(\varepsilon^0)$, $z^k_{t_1} = - z^{-k}_{t_1}$ for any value of $k \neq 0$. Moreover, as the lattice equation for $v^0_n$ from \eqref{eq:t1v} depends only on the site $n$, the corresponding equation for $z^{0}$ does not carry higher order correction in $\varepsilon$.  From the evolution equations \eqref{eq:t1w} for $w^k_n$  we find 
\begin{gather}
\label{t1flowsbis}
\begin{aligned}
  u^{k}_{t_1} =& \,2 z^0 u^{k} + u^0 \left(z^{k+2}-z^{-(k+2)}+z^{k} - z^{-k}\right) + u^{-1} \left(z^{k+1} - z^{-(k+1)}\right) \\[1.5ex]
    &+ \varepsilon\left(-(k+1)(z^{k+1} u^{-1}_x + z^{k} u^0_x)-(k+2)( z^{k+2}u^0_x + u^{k}z^0_x)\right. \\[1.5ex]
    &\hspace{6ex}\left. - u^0 z^{-(k+2)}_x + (u^0 z^{-k})_x\right)+\frac{\varepsilon^2}{2} \left((k+1)^2 (z^{k+1}_x u^{-1}_{xx}+z^{k} u^0_{xx}) \right.\\[1.5ex]
    &\left.\hspace{6ex} +(k+2)^2 z^{(k+2)}u^0_{xx}+ (3 + k(k+3))u^{k}z^0_{xx}- u^0 z^{-(k+2)}_{xx} -(u^0 z^{-k})_{xx}\right)\\[1.5ex]
    &+ O(\varepsilon^3),\hspace{10ex}k<-1\\[1.5ex]
    u^{-1}_{t_1} =& \,u^0 (z^{-1} - z^1) + \varepsilon\left(u^0 z^1\right)_x -\frac{\varepsilon^2}{2}\left(u^0z^1\right)_{xx} + O(\varepsilon^3)\\[1.5ex]
    u^0_{t_1} =& \,\frac{\varepsilon^2}{2} z^0_{xx} u^0 + O(\varepsilon^4)\\[1.5ex]
    u^k_{t_1} =& -2z^0 u^k + z^k - z^{-k} - \varepsilon(k-1)z^0_x u^k - \frac{\varepsilon^2}{2}(1 + k(k-1))u^k z^0_{xx} + O(\varepsilon^3),\;\;\;\;k\geq1.
\end{aligned}
\end{gather}
Notice that equations \eqref{t1flows} and \eqref{t1flowsbis}, unlike their counterpart in the reduced even Pfaff Lattice \eqref{hydrochain}, are not quasilinear and do not constitute a hydrodynamic chain.

\section{Continuum limit of the even Pfaff Hierarchy: \texorpdfstring{$t_{2}-$}{t2}flow and higher order corrections}\label{app:even_corrections}
We provide, for the reduced even Pfaff Hierarchy, the corrections to the leading order of equation \eqref{hydrochain}  up to  $O(\varepsilon^3)$:
\begin{gather}
    \begin{aligned}
        u^k_t & =  \left(\left((k+2)  u^{k+1}-k  u^{k-1}-u^{1}u_x^{0} u^{k}\right)u_x^{0}-u^{0} u_x^{1} u^{k}+u^{0} u_x^{k-1}+u^{0} u_x^{k+1}\right)  \\[1.5ex] 
        & \hspace{2ex} +\frac{1}{2}  \left(k^2 u_{xx}^{0} (-u^{k-1})+(k^2+2k) u_{xx}^{0} u^{k+1}-k u^{k} \left(2 u_x^{0} u_x^{1}+u^{1} u_{xx}^{0}+u^{0} u_{xx}^{1}\right) \right. \\[1.5ex] 
        & \left. \hspace{6ex} -2 u_x^{0} u_x^{k+1}+u^{0}\left( u_{xx}^{k-1}- u_{xx}^{k+1}\right)\right)  \varepsilon\\[1.5ex]
        &\hspace{2ex}+\frac{1}{12}  \left(2 \left( u_{xxx}^{0} \left( ((k+1)^3+1) u^{k+1}-k^3 u^{k-1}\right)+3 u_{xx}^{0} u_x^{k+1}+3 u_x^{0} u_{xx}^{k+1} \right. \right.\\[1.5ex]
        & \left. \hspace{6ex}+u^{0} \left( u_{xxx}^{k-1}+ u_{xxx}^{k+1}\right) \right) \\[1.5ex]
        & \left. \hspace{6ex}- u^{k}\left(3 k^2+3 k+2\right) \left(3 u_x^{1} u_{xx}^{0}+3 u_x^{0} u_{xx}^{1}+u^{1} u_{xxx}^{0}+u^{0} u_{xxx}^{1}\right)\right) \varepsilon ^2\\[1.5ex]
        &\hspace{2ex}+O\left(\epsilon ^3\right) \,, \hspace{10ex} k<0 \\[1.5ex]
        u^0_t & = u^{0} \left(u^{-1}_x+u^{1} u^{0}_x+u^{0} u^{1}_x\right)+\frac{1}{2} \left( u^{0} u^{-1}_{xx} \right) \, \varepsilon \\[1.5ex] & \hspace{2ex} +\frac{1}{6}  u^{0} \left(3 u^{1}_x u^{0}_{xx}+3 u^{0}_x u^{1}_{xx}+u^{-1}_{xxx}+u^{1} u^{0}_{xxx}+u^{0} u^{1}_{xxx}\right) \varepsilon^2 +O\left(\varepsilon ^3\right) \\[1.5ex] 
        u^1_t &=  \left(2 u^2 u_x^0-u^1 \left(u^1 u_x^0+u^0 u_x^1\right)+u^0 u_x^2\right)  +\left(-u_x^0 u_x^2-\frac{1}{2} u^0 u_{xx}^2\right)\,\varepsilon \\[1.5ex]
        & \hspace{2ex} +\frac{1}{6}  \left(-u_{xxx}^0 (u^1)^2-\big(3 u_x^1 u_{xx}^0+3 u_x^0 u_{xx}^1+u^0 u_{xxx}^1\big) u^1 \right.\\[1.5ex] 
        & \left. \hspace{6ex}+3 u_x^2 u_{xx}^0+3 u_x^0 u_{xx}^2+2 u^2 u_{xxx}^0+u^0 u_{xxx}^2\right) \,\varepsilon^2+O\big(\varepsilon^3\big) \\[1.5ex]
        u^{k}_t&= \left(\left((k+1)  u^{k+1}-(k-1)  u^{k-1}+u^1 u^k \right) u_{x}^0 +u^0 u_{x}^1 u^k+u^0 u_{x}^{k-1}+u^0 u_{x}^{k+1}\right)  \\[1.5ex] & \hspace{2ex} +\frac{1}{2} \left(u_{xx}^0\big((k^2-1) u^{k+1}-(k^2-2k+1)u^{k-1}\big)- 2 u_{x}^0 u_{x}^{k+1} \right. \\[1.5ex] 
        & \left. \hspace{6ex} +(k-1) \left(2 u_{x}^0 u_{x}^1+u^1 u_{xx}^0+u^0 u_{xx}^1\right) u^k+u^0 u_{xx}^{k-1}-u^0 u_{xx}^{k+1}\right) \varepsilon   \\[1.5ex] 
        & \hspace{2ex} +\frac{1}{12} \left(2 \left( u_{xxx}^0 \left((k^3+1)u^{k+1}-(k-1)^3u^{k-1}\right)+3 u_{xx}^0 u_{x}^{k+1} \right. \right. \\[1.5ex] 
        & \left. \hspace{6ex}+3 u_{x}^0 u_{xx}^{k+1}  +u^0 \left( u_{xxx}^{k-1}+u_{xxx}^{k+1} \right) \right)+\big(3 k^2-3 k+2\big) \left(3 u_{x}^1 u_{xx}^0 \right. \\[1.5ex]
        &\left. \left. \hspace{6ex}+3 u_{x}^0 u_{xx}^1+u^1 u_{xxx}^0+u^0 u_{xxx}^1\right) u^k\right)  \varepsilon ^2+O\left(\varepsilon ^3\right) \,, \hspace{10ex} k>1
    \end{aligned}
\end{gather}

\newpage
\section{The Nijenhuis tensor}\label{sec:N}
We list the explicit form of the nonzero elements of the Nijenhuis tensor $N^i_{jk}$ for any value of $i$. They are evaluated directly based on the definition of the tensor \eqref{eq:Ndef} and the form of the matrix $A(\mathbf{u})$ given by the expressions \eqref{eq:a}.  As mentioned in the proof of Proposition \ref{prop:diagonalisability}, the only nonzero elements of $N^i_{jk}$ are the ones listed in  \eqref{eq:N_nonzero} along with their counterparts with lower indices exchanged -- recall that $N^i_{jk} = - N^i_{kj}$.  \\

For $|i|>2$
\begin{equation}
\begin{split}
    N^i_{0\,1}&= \left\{\begin{array}{ll}u^0\left((i-1)u^{i-1} - (i+1)u^{i+1}\right)& \mbox{ if }i>2\\[1.2ex]
    u^0\left(i u^{i-1} - (i+2) u^{i+1}\right)&\mbox{ if }i < -2\end{array}\right.\\[1.2ex]
        N^i_{0\,-1} &= \left\{\begin{array}{ll}(i-1)u^{i-1} + u^1 u^i - (i+1)u^{i+1}& \mbox{ if }i> 2\\[1.2ex]
   i u^{i-1} - u^i u^1-(i+2) u^{i+1}&\mbox{ if }i< -2\end{array}\right.
    \end{split}\nonumber
\end{equation}

\begin{minipage}[h]{0.45 \linewidth}
\centering
\begin{equation}
    \nonumber
    \begin{split}
            &N^i_{-1,1} = -\mbox{sgn}(i) u^0 u^i \hspace{6ex} \\
            &N^i_{0,i} = - 4 u^0\\
            &N^i_{0,i+1} = u^0 u^1\\
            &N^i_{0,i-1} = u^0 u^1
    \end{split}
\end{equation}
\vspace{0.2cm} 
\end{minipage}
\begin{minipage}[h]{0.45 \linewidth}
\centering
\begin{equation}
    \nonumber
    \begin{split}
             &N^i_{1,i+1} = (u^0)^2 \hspace{23ex} \\
             &N^i_{1,i-1} = (u^0)^2\\
            &N^i_{-1,i+1} = u^0\\
           &N^i_{-1,i-1} = u^0
    \end{split}
\end{equation}
\vspace{0.2cm} 
\end{minipage}

For $|i|\leq 2$

\begin{minipage}[h]{0.45\linewidth}
\begin{equation}
\begin{split}
&N^2_{0,1}=u^0 (2 u^1 - 3 u^3)\\
&N^2_{0,-1}=u^1 (1 + u^2) - 3 u^3\\
&N^2_{-1,1}=-u^0 (-1 + u^2)\\
&N^2_{0,2}= - 4 u^0\\
&N^2_{0,3} = u^0 u^1\\
&N^2_{1,3}=(u^0)^2\\
&N^2_{-1,3}=u^0\\
&N^1_{0,1}=-2 u^0 (2 + u^2)\\
&N^1_{0,2}= u^0 u^1\\
&N^1_{1,2}=(u^0)^2\\
&N^1_{-1,0}= - (u^1)^2 + 2 u^2 \\
&N^1_{-1,1}=-u^0 u^1\\
&N^1_{-1,2}=u^0\nonumber
\end{split}
\end{equation}
\end{minipage}
\begin{minipage}[h]{0.45 \linewidth}
\begin{equation}
\begin{split}
&N^{-2}_{0,1}=-2 u^{-3} u^0\\
&N^{-2}_{0,-1}=-2 u^{-3} + (-u^{-2} + u^{0}) u^{1}\\
&N^{-2}_{-1,1}=(u^{-2} - u^{0}) u^{0}\\
&N^{-2}_{0,-2}= - 4 u^0\\
&N^{-2}_{0,-3} = u^0 u^1\\
&N^{-2}_{1,-3}=(u^0)^2\\
&N^{-2}_{-1,-3}=u^0\\
&N^{-1}_{0,1}=-u^0 (u^{-2} + 2 u^0)\\
&N^{-1}_{0,-1}= -u^{-2} - 6 u^{0} - u^{-1} u^1\\
&N^{-1}_{0,-2}=u^0 u^1\\
&N^{-1}_{1,-2}= (u^0)^2\\
&N^{-1}_{-1,-2} = u^0 \\
&N^{-1}_{-1,1}=u^0 u^{-1}\nonumber
\end{split}    
\end{equation}
\end{minipage}

\end{document}